\documentclass[letterpaper,10pt]{article}
\usepackage[letterpaper,margin=1in]{geometry}
\usepackage{url}

\title{Behavioral Mapping of Hierarchical Dynamical Models Accelerates Bayesian Parameter Inference}
\author{Jingsong Bai and Ayush Pandey%
\thanks{The authors are with the Department of Electrical Engineering and Computer Science,
University of California, Merced. Correspondence should be addressed to
Ayush Pandey (\texttt{ayushpandey@ucmerced.edu}).}%
}
\date{}
\usepackage{mathtools,amssymb,amsfonts}
\usepackage{mathrsfs}
\usepackage{algorithm}
\usepackage{algorithmic}
\usepackage{booktabs}
\usepackage{graphicx}
\usepackage{array}
\usepackage{textcomp}
\usepackage{xcolor}
\usepackage{float}
\usepackage{placeins}
\usepackage{arydshln}
\usepackage{cite}
\usepackage{accents}
\usepackage{relsize}

\def\BibTeX{{\rm B\kern-.05em{\sc i\kern-.025em b}\kern-.08em
    T\kern-.1667em\lower.7ex\hbox{E}\kern-.125emX}}

\newcommand{\norm}[1]{\left\lVert#1\right\rVert}

\usepackage{amsthm}

\newtheorem{theorem}{Theorem}
\newtheorem{corollary}{Corollary}[theorem]

\newtheorem{proposition}{Proposition}

\theoremstyle{definition}
\newtheorem{definition}{Definition}
 
\theoremstyle{remark}

\theoremstyle{remark}

\newcommand{\GeneMCMCTimeReductionPct}{15.21}

\newcommand{\GeneBaselineMinutes}{316.91}
\newcommand{\GeneGuidedMinutes}{268.71}
\newcommand{\GeneReducedMinutes}{6.64}
\newcommand{\GenePhenomenologicalMinutes}{92.38}
\newcommand{\GeneGPTrainingMinutes}{4.46}
\newcommand{\GeneReducedRMSE}{0.255}
\newcommand{\GenePhenomenologicalRMSE}{0.272}
\newcommand{\GeneBaselineRMSE}{0.327}
\newcommand{\GeneGuidedRMSE}{0.309}
\newcommand{\GeneReducedWalkers}{64}
\newcommand{\GeneReducedSteps}{10000}
\newcommand{\GenePhenomenologicalWalkers}{256}
\newcommand{\GenePhenomenologicalSteps}{40000}
\newcommand{\GeneDetailedWalkers}{256}
\newcommand{\GeneDetailedSteps}{40000}
\newcommand{\GeneDetailedBurn}{5000}

\newcommand{\GeneGPBalancedErrorPct}{0.94}
\newcommand{\GeneGPStrictBalancedErrorPct}{1.89}
\newcommand{\GeneRandomBalancedErrorPct}{2.57}
\newcommand{\GeneGPCompatibleRetainedPct}{96.23}
\newcommand{\GeneGPIncompatibleRejectedPct}{100.00}
\newcommand{\GeneGPProposalRejectedPct}{12.69}
\newcommand{\GeneGPScreeningMinutes}{1.79}
\newcommand{\GeneDetailedPairFirst}{$k_{tl}$}
\newcommand{\GeneDetailedPairSecond}{$k_{dil}$}
\newcommand{\GeneReducedPairFirst}{$k_{tl}$}
\newcommand{\GeneReducedPairSecond}{$k_{dil}$}

\newcommand{\GeneHighLevelFit}{\includegraphics[width=\linewidth]{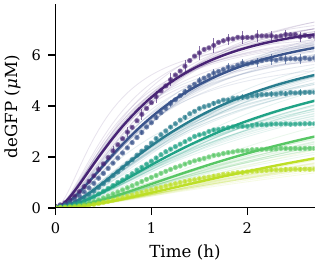}}
\newcommand{\GenePhenomenologicalFit}{\includegraphics[width=\linewidth]{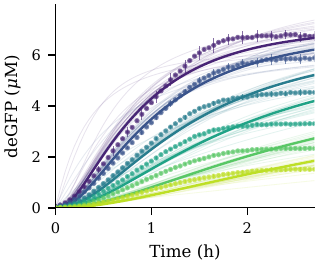}}
\newcommand{\GeneDetailedFit}{\includegraphics[width=\linewidth]{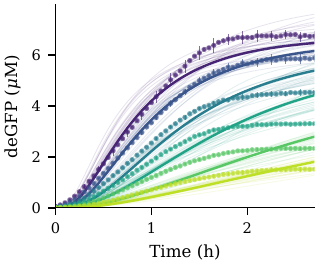}}
\newcommand{\GeneHierarchyLegend}{\includegraphics[width=\linewidth]{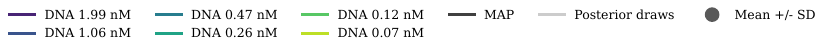}}
\newcommand{\GeneParameterDependence}{\includegraphics[width=0.8\linewidth]{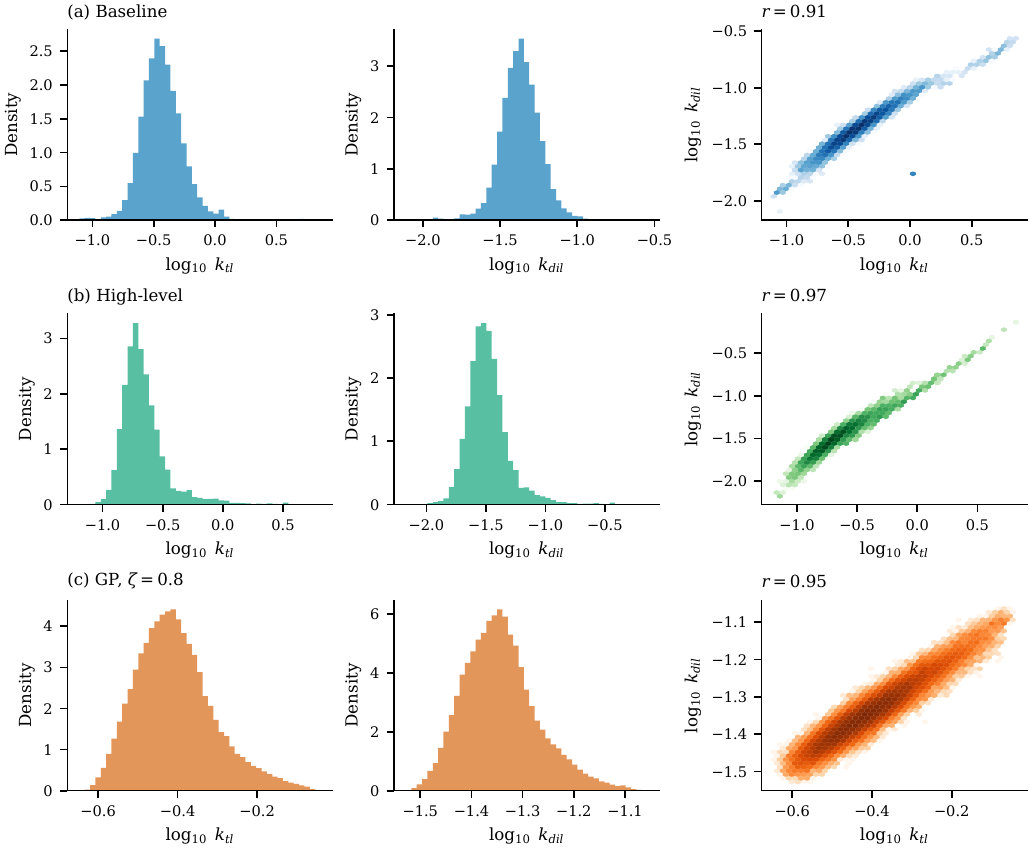}}
\newcommand{\GeneGPTrainingScreening}{\includegraphics[width=0.8\linewidth]{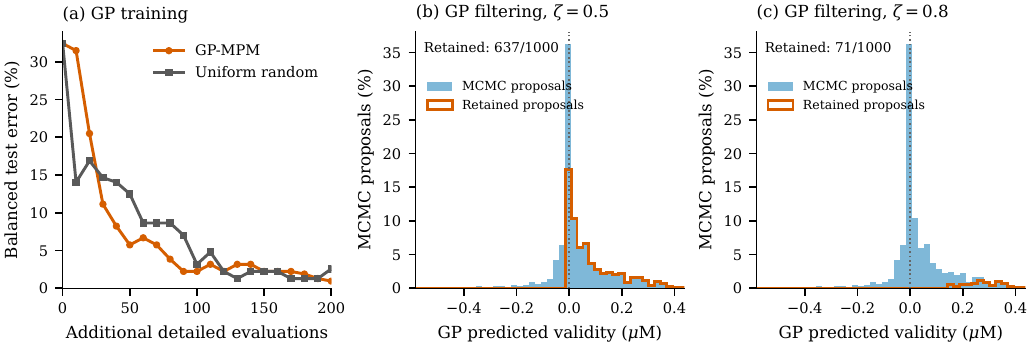}}
\newcommand{\GeneBaselineGuidedFits}{\includegraphics[width=\linewidth]{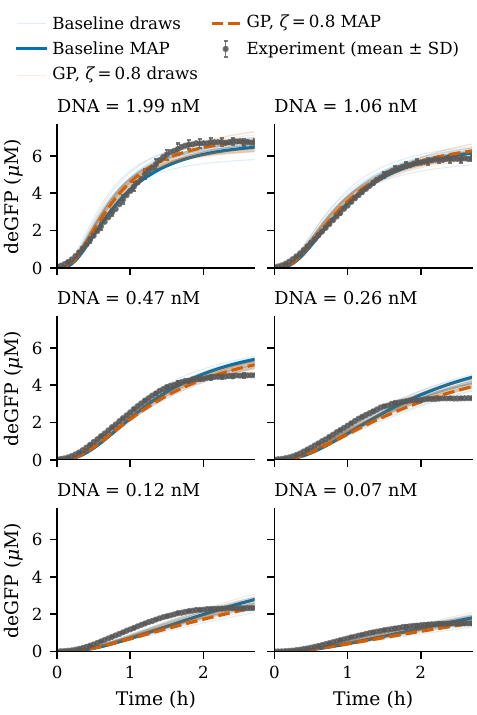}}

\newcommand{\GeneHPMCMCTimeReductionPct}{8.02}

\newcommand{\GeneHPBaselineMinutes}{92.38}
\newcommand{\GeneHPGuidedMinutes}{84.98}
\newcommand{\GeneHPReducedMinutes}{6.64}
\newcommand{\GeneHPGPTrainingMinutes}{2.94}

\newcommand{\GeneHPBaselineRMSE}{0.272}
\newcommand{\GeneHPGuidedRMSE}{0.361}

\newcommand{\GeneHPGPBalancedErrorPct}{8.22}
\newcommand{\GeneHPGPStrictBalancedErrorPct}{26.46}
\newcommand{\GeneHPRandomBalancedErrorPct}{12.30}
\newcommand{\GeneHPGPCompatibleRetainedPct}{47.73}
\newcommand{\GeneHPGPIncompatibleRejectedPct}{99.36}
\newcommand{\GeneHPGPProposalRejectedPct}{31.66}
\newcommand{\GeneHPGPScreeningMinutes}{2.01}

\newcommand{\GeneHPGPTrainingScreening}{\includegraphics[width=0.8\linewidth]{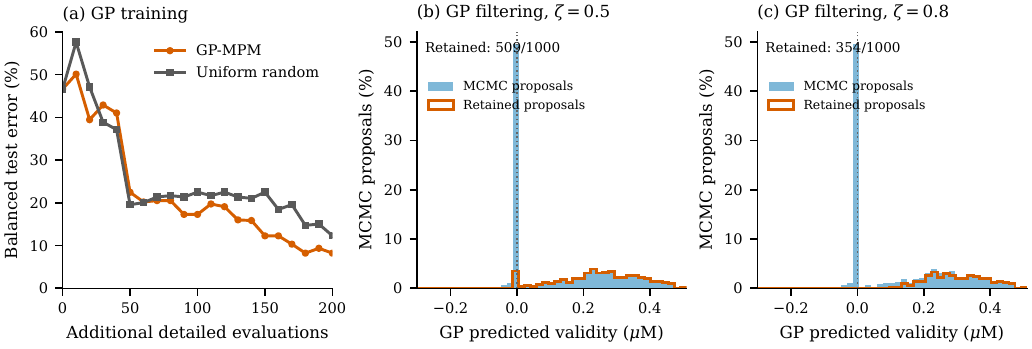}}
\newcommand{\GeneHPBaselineGuidedFits}{\includegraphics[width=\linewidth]{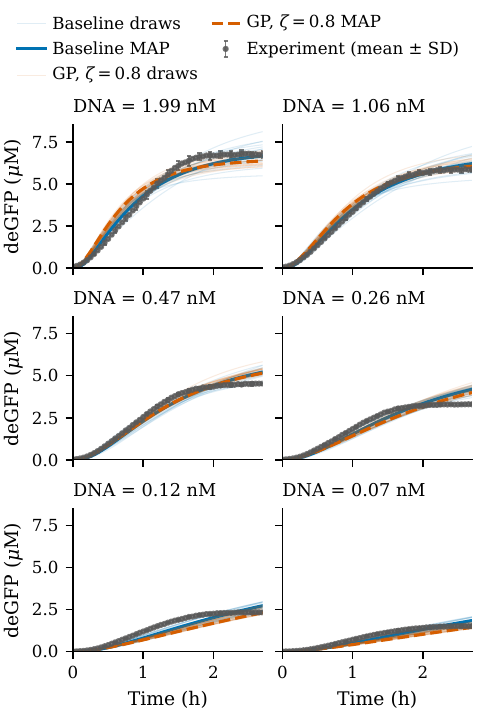}}
\newcommand{\GeneHierarchyMarginals}{\includegraphics[width=\linewidth]{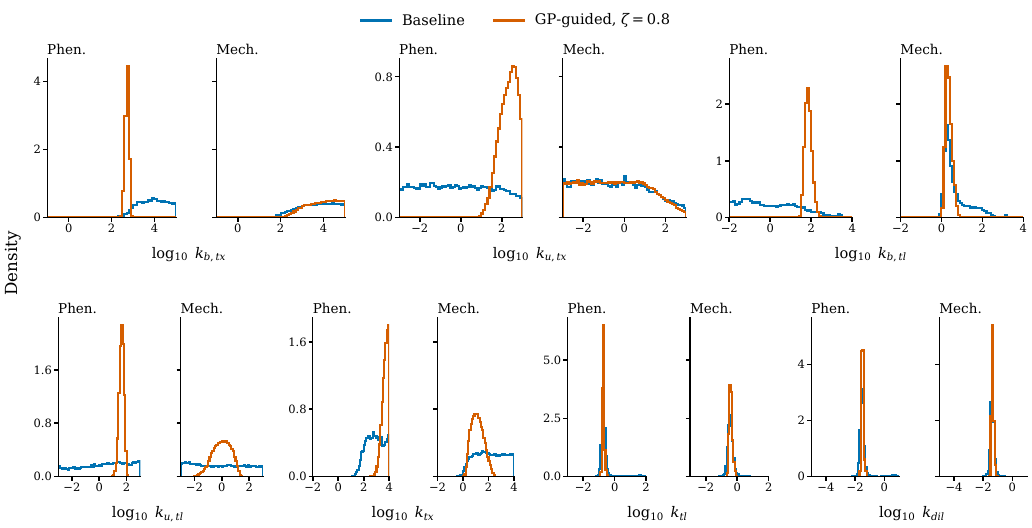}}

\newcommand{\GeneHighLevelCRN}{%
\[
\begin{aligned}
\varnothing &\xrightarrow{v_{\text{tx}}} \text{T}, &
\text{T} &\xrightarrow{v_{\text{deg}}}\varnothing,\\
\text{T} &\xrightarrow{v_{\text{tl}}} \text{T}+\text{X}, &
\text{X} &\xrightarrow{k_{\text{dil}}}\varnothing.
\end{aligned}
\]
\smallskip
{\centering AutoReduce dynamics\par}
\[
\begin{aligned}
\dot{\text{T}} &= k_{\text{tx}}\text{P}_{\text{tot}}\frac{\text{G}}{\text{G}+K_1}\\
       &\quad-k_{\text{deg}}\text{E}_{\text{tot}}\frac{\text{T}}{K_2+\text{T}},\\
\dot{\text{X}} &= k_{\text{tl}}\text{R}_{\text{tot}}\frac{\text{T}}{K_0+\text{T}}-k_{\text{dil}}\text{X}.
\end{aligned}
\]}

\newcommand{\GenePhenomenologicalCRN}{%
\[
\begin{aligned}
\text{G}+\text{P} &\mathrel{\mathop{\rightleftharpoons}^{k_{\text{b},\text{tx}}}_{k_{\text{u},\text{tx}}}} \text{C}_{\text{tx}}
\xrightarrow{k_{\text{tx}}}\text{G}+\text{P}+\text{T},\\[2pt]
\text{T}+\text{R} &\mathrel{\mathop{\rightleftharpoons}^{k_{\text{b},\text{tl}}}_{k_{\text{u},\text{tl}}}} \text{C}_{\text{tl}}
\xrightarrow{k_{\text{tl}}}\text{T}+\text{R}+\text{X},\\[2pt]
\text{T}+\text{E} &\mathrel{\mathop{\rightleftharpoons}^{k_{\text{b},\text{deg}}}_{k_{\text{u},\text{deg}}}} \text{C}_{\text{deg}}
\xrightarrow{k_{\text{deg}}}\text{E},\\[2pt]
\text{X} &\xrightarrow{k_{\text{dil}}}\varnothing.
\end{aligned}
\]}

\newcommand{\GeneDetailedCRN}{%
\[
\begin{aligned}
\text{S}+\text{A} &\mathrel{\mathop{\rightleftharpoons}^{k_{\text{b},r}}_{k_{\text{u},r}}} \text{C}^{\text{c}}_{1,r},\\
\text{C}^{\text{o}}_{n,r}+\text{A} &\mathrel{\mathop{\rightleftharpoons}^{k_{\text{b},r}}_{k_{\text{u},r}}} \text{C}^{\text{c}}_{n+1,r},\\
\text{C}^{\text{c}}_{n,r} &\xrightarrow{k_{\text{iso},r}} \text{C}^{\text{o}}_{n,r},\\
\text{C}^{\text{o}}_{n,r} &\xrightarrow{k_r}\text{S}+n\text{A}+n\text{V},\\
\text{C}^{\text{c}}_{n,r} &\xrightarrow{k_r}\text{C}^{\text{c}}_{1,r}+(n-1)(\text{A}+\text{V}),\\
\text{X}_{\text{u}} &\xrightarrow{k_{\text{mat}}}\text{X}.
\end{aligned}
\]}

\newcommand{\GeneHierarchyPanel}[1]{%
\begin{minipage}[t]{\linewidth}
\vspace{0pt}\fontsize{8}{10}\selectfont
\setlength{\abovedisplayskip}{3pt}
\setlength{\belowdisplayskip}{3pt}
\setlength{\abovedisplayshortskip}{3pt}
\setlength{\belowdisplayshortskip}{3pt}
#1
\end{minipage}}

\newcommand{\GeneHighLevelMetrics}{%
\begin{tabular}{@{}lr@{}}
MCMC time (min) & \GeneReducedMinutes\\
Mean MAP RMSE ($\mu$M) & \GeneReducedRMSE\\
Walkers $\times$ steps & \GeneReducedWalkers$\times$\GeneReducedSteps
\end{tabular}}
\newcommand{\GenePhenomenologicalMetrics}{%
\begin{tabular}{@{}lr@{}}
MCMC time (min) & \GenePhenomenologicalMinutes\\
Mean MAP RMSE ($\mu$M) & \GenePhenomenologicalRMSE\\
Walkers $\times$ steps & \GenePhenomenologicalWalkers$\times$\GenePhenomenologicalSteps
\end{tabular}}
\newcommand{\GeneDetailedMetrics}{%
\begin{tabular}{@{}lr@{}}
MCMC time (min) & \GeneBaselineMinutes\\
Mean MAP RMSE ($\mu$M) & \GeneBaselineRMSE\\
Walkers $\times$ steps & \GeneDetailedWalkers$\times$\GeneDetailedSteps
\end{tabular}}

\begin{document}
\maketitle
\pagestyle{plain}

\begin{abstract}
Bayesian parameter inference for complex nonlinear dynamical models is often limited by the cost of repeated model simulation. We present a behavioral mapping framework that uses the inferred posterior of a higher-level model to accelerate inference of a lower-level model without requiring state or parameter correspondence between them. Building on approximate simulation relations for hierarchically related models, we define the notion of behavioral compatibility and quantify the behavioral support of detailed model parameters using the reduced model posterior. We learn a level set of behavioral support using a Gaussian process and use it to filter detailed model parameter proposals before expensive simulations during Markov Chain Monte Carlo (MCMC) inference. We derive conditions under which the behavioral support level set retains all detailed model parameters that fit the experimental data within a prescribed tolerance. We apply this method to a hierarchy of gene-expression models fitted to cell-free experimental data. After initial setup costs, our approach speeds up MCMC by \GeneHPMCMCTimeReductionPct\% for the phenomenological model and \GeneMCMCTimeReductionPct\% for the mechanistic model.
\end{abstract}

\section{Introduction}
To use a dynamical model for design or analysis of a physical system, it is necessary to first identify the parameters of the model from experimental data so that the model can accurately predict the system behavior. In practice, parameter inference is often computationally expensive to the extent that even for moderately sized models, it can take anywhere between days to weeks to obtain reasonable estimations or probability distributions of parameter values. This is especially true for nonlinear ordinary differential equation (ODE) models that are commonly used in the modeling of biological, chemical, and physical systems. These models usually do not admit closed-form solutions so analytical methods are not available for parameter estimation, the optimization problems are non-convex, and the parameters are often highly correlated leading to non-identifiability. So, parameter inference methods must rely on full simulations of the model. A commonly used parameter inference method for nonlinear ODEs is Bayesian inference, which estimates the posterior probability distribution of the parameters given the data from observations of the system. Approaches like Approximate Bayesian Computation (ABC), Sequential Monte Carlo (SMC), and Markov Chain Monte Carlo (MCMC) are used to sample from the posterior distribution of the parameters. With these methods, the computational cost of parameter inference often increases exponentially with the number of parameters in the model. Therefore, there is a severe limitation on the size of the models that can be used for parameter inference, and thus, on our ability to use models accurately to predict complex system behavior. To address this limitation, in this paper, we propose a new approach that learns information about system behavior from hierarchically related reduced models to accelerate the parameter inference of the full model. 

% epresentative modeling hierarchy:
%         behavioral specification
%         phenomenological dynamical model
%         detailed mechanistic model

\subsection{Related Work}
Our approach is at the intersection of several recent developments in reduced-order modeling, model validation, Bayesian methods, and uncertainty quantification.
% We discuss this related literature briefly to appropriately motivate our approach.  
% Approximate simulation and bisimulation relations have been used to formally quantify how closely the external behavior of one dynamical system can reproduce another, particularly for abstraction, model reduction, and verification~\cite{girard2007approximation}. Their use in parameter inference, however, has been limited, especially for transferring information between models with different parameterizations. Our work is in line with 
Prior work has shown that low-fidelity models can be used to accelerate parameter inference, uncertainty propagation, or optimization, while high-fidelity models provide accuracy or convergence control~\cite{peherstorfer2018survey}. Specifically for Bayesian inference, delayed acceptance is an approach that has been used to evaluate the computationally expensive posterior target only after a proposal has passed a first-stage acceptance criterion~\cite{christen2005mcmc}. Transport maps~\cite{peherstorfer2019transport} from low-fidelity posteriors have also been used to propose samples for a high-fidelity posterior. On similar lines, projection-based reduced-order modeling has been adaptively used~\cite{cui2015datadriven} during posterior exploration to couple the reduced and full models within the MCMC routine. These approaches implicitly require the same parameter coordinates or a correspondence between the low-fidelity and high-fidelity model parameters, thus limiting their applicability to cases where such a correspondence exists or can be easily derived. Recently, such approaches have also been extended to support arbitrary model hierarchies~\cite{lykkegaard2023mlda,zacchei2026mfda}, while still requiring related parameterizations or a known correspondence between model levels. 

Methods that do not depend on parameter mapping either use a statistical response relation between models to correct the likelihood computation in gradient-based inference~\cite{nitzler2026bmfia} or adaptively learn a polynomial or Gaussian Process (GP) approximation of the target forward model locally~\cite{conrad2016local,li2014adaptive} within MCMC inference. These approaches use detailed model simulations to learn either the cross-fidelity response relation or the local approximation of the detailed model, which can itself constitute a substantial computational bottleneck.
From an algorithmic perspective, simulated-to-observed data discrepancy has been modeled with a GP in ABC to reject proposals before expensive simulation~\cite{cao2025early}. This is algorithmically similar to our use of GP filtering, but the GP models target model-to-data discrepancy, whereas our GP learns a behavioral support boundary induced by the posterior of a higher-level model and is used to filter proposals in the MCMC sampling for the lower-level detailed model.
Although reduced-order models have been widely used to support and speed up parameter inference their use typically relies on an explicit mapping or otherwise imposed structural constraints. Here, we formalize how information from the posterior distribution of an inferred model at a higher level of abstraction can act as behavioral support over the parameter space of the detailed model at a lower level of abstraction, even when no parameter or state correspondence between the models is available. This information is transferred through the posterior probability of finite-horizon output compatibility. For this, we use approximate simulation relations to certify the finite-horizon compatibility tolerance between models at different levels of the hierarchy~\cite{girard2007approximation}, and a recently developed GP minimum posterior misclassification (GP-MPM) method~\cite{moorse2024validation} to efficiently learn the induced parameter-space support level set. This learned level set is then used to filter MCMC proposals before expensive model simulation, thereby accelerating inference.

\subsection{Contributions}
The main contributions of this paper are summarized below:
\begin{enumerate}
    \item We present a computationally efficient framework for defining and learning the behavioral compatibility between hierarchically related models, without requiring state or parameter mapping between the models. 
    \item We construct a behavioral support level set on the parameter space of a model and derive conditions under which this set retains the desired data-fitting region, enabling computationally efficient MCMC inference.
    \item For nonlinear gene expression, we accelerate MCMC by \GeneHPMCMCTimeReductionPct\% for a nine-state phenomenological model and \GeneMCMCTimeReductionPct\% for a 38-state mechanistic model after initial setup costs, and compare the inferred parameter distributions and GFP fits with baseline inference.
\end{enumerate}

\section{Preliminaries}
Let $k$ denote the common dimension of the measured outputs of the detailed
and reduced models. In this paper, we consider the Hilbert space
\(
    \mathcal{H}_T
    :=
    L_2([0,T];\mathbb{R}^k)
\)
over a finite time horizon $[0,T]$ with norm
\[
    \norm{y(t)}_{\mathcal{H}_T} = 
    \left(
        \int_0^T \|y(t)\|_2^2\,dt
    \right)^{1/2}.
\]
For brevity, we will drop the $T$ notation in what follows. The induced finite-horizon distance metric we use is
\(
    d(y_1,y_2)
    :=
    \|y_1-y_2\|_{\mathcal{H}}.
\)
More generally, the results above apply to any finite metric on a common measured output space, provided the parameterized discrepancies are measurable. Below, we use the finite-horizon \(L_2\) metric.
We consider a hierarchy of models that are ordered by complexity --- high-level abstractions capture the system-level specification behavior while low-level models capture the detailed dynamics. For each pair in the hierarchy, we have a detailed model and a reduced-order model. We do not assume any ordering for parameters or \textit{a priori} knowledge of the relationship between them. For both levels, in the most general case, we consider well-posed nonlinear parameterized models: \begin{align}
\dot{x} = f(x, u, \theta), \: y = Cx \label{eq:nl-sys}\\ \dot{\hat{x}} = \hat{f}(\hat{x}, u, \hat{\theta}),\: \hat y = \hat C \hat x 
\label{eq:nl-sys-reduced}
\end{align} 
where $x \in \mathbb{R}^n, \hat x \in \mathbb{R}^{\hat n}$ are the states, $u \in \mathbb{R}^m$ are the inputs, and $\theta \in \Theta\subseteq\mathbb{R}^p, \hat\theta \in \widehat\Theta\subseteq\mathbb{R}^{\hat p}$ are the parameters. Both models describe the same measured outputs, with $y,\hat y\in\mathbb{R}^{k}$, $C\in\mathbb{R}^{k\times n}$, and $\hat C\in\mathbb{R}^{k\times\hat n}$. The outputs depend linearly on the states. The $\hat{\cdot}$ counterparts represent the reduced-order model equivalents. We write the parameter dependence explicitly as $y(\cdot;\theta)$.

\begin{definition}[Directed Hausdorff distance]
    Let
\begin{equation}
    \mathcal{Y}
    :=
    \{y(\cdot;\theta):\theta\in\Theta\},
    \qquad
    \widehat{\mathcal{Y}}
    :=
    \{\hat y(\cdot;\hat\theta):\hat\theta\in\widehat\Theta\}.
    \label{eq:behavior-sets}
\end{equation}
Then, the directed Hausdorff distance from the detailed model behavior set to the
reduced model behavior set is
\begin{equation}
    d_H(\mathcal{Y},\widehat{\mathcal{Y}})
    :=
    \sup_{y\in\mathcal{Y}}
    \inf_{\hat y\in\widehat{\mathcal{Y}}}
    d(y,\hat y).
    \label{eq:hausdorff-distance}
\end{equation}
Equivalently,
\[
    d_H(\mathcal{Y},\widehat{\mathcal{Y}})
    =
    \sup_{\theta\in\Theta}
    \inf_{\hat\theta\in\widehat\Theta}
    d\!\left(
        y(\cdot;\theta),
        \hat y(\cdot;\hat\theta)
    \right).
\]
\end{definition}
% This quantity measures the worst-case finite-horizon error with which the
% reduced model family can represent a detailed-model behavior.
For linear autonomous systems ($\dot x = Ax$, $y=Cx$ and $\dot{\hat x}=\hat A\hat x$, $\hat y=\hat C\hat x$), we note the following finite-horizon observability Gramians
\[
    Q_T
    :=
    \int_0^T
    e^{A^\top t}C^\top C e^{At}\,dt, \: 
    \hat Q_T
    :=
    \int_0^T
    e^{\hat A^\top t}\hat C^\top \hat C e^{\hat A t}\,dt,
\]
and the cross-term Gramian
\[
    H_T
    :=
    \int_0^T
    e^{A^\top t}C^\top \hat C e^{\hat A t}\,dt.
\]
For the specialized linear analysis below, we assume that the reduced pair
$(\hat A,\hat C)$ is observable over $[0,T]$. Consequently,
$\hat Q_T\succ0$ and $\hat Q_T$ is invertible.
\begin{definition}[Approximate simulation~\cite{girard2007approximation}]
A relation $S_\epsilon \subset \mathbb{R}^{n}\times\mathbb{R}^{\hat n}$ is an approximate simulation relation of the detailed
system by the reduced system over $[0,T]$ if three conditions hold: (1) Bounded observations for every $(x,\hat x)\in S_\epsilon$,
    \[
        \|y-\hat y\|_2
        \le
        \frac{\epsilon}{\sqrt{T}}.
    \]
(2) Transition matching: For every $(x,\hat x)\in S_\epsilon$, every transition of the detailed
    state from $x$ can be matched by a transition of the reduced state from
    $\hat x$ to a new pair that remains in $S_\epsilon$, and (3) Coverage of the initial condition set which says that every detailed initial condition under consideration has at least one
    related reduced initial condition. Note that we denote the desired finite-horizon compatibility tolerance as $\epsilon$ and we use the corresponding instantaneous output precision $\epsilon/\sqrt{T}$. For deterministic continuous-time flows, transition matching implies that a
related initial-state pair remains related over the experiment horizon.
\end{definition}

\subsection{Experimental data and desired fit tolerance}

Let $\bar y\in\mathcal{H}_T$ denote the measured output trajectory and $d(\cdot,\cdot)$ denote the metric induced by the norm on $\mathcal{H}$. We denote $\eta>0$ as the desired model-data fit tolerance. We define the set
\begin{equation}
    F_\eta
    :=
    \left\{
        \theta\in\Theta:
        d\!\left(
            y(\cdot;\theta),\bar y
        \right)
        \le
        \eta
    \right\}
\label{eq:data-fit-set}
\end{equation}
for the parameter values of the detailed model that produce outputs that fit the measured data within the tolerance. Similarly, for the model of reduced complexity, for $s\ge0$, we define $\widehat F_s$ as the set of $\hat \theta \in \hat{\Theta}$ such that $d\!\left(\hat y(\cdot;\hat\theta),\bar y\right)\le s$.
\subsection{Bayesian inference and MCMC notation}
We denote $\pi_0(\theta)$ as a prior density on model parameters $\theta$ and $\pi(\theta\mid\bar y)$ as the posterior density. Additionally, we define the posterior probability mass within $F_s$ as
\[
    m_{\pi}(\eta)
    :=
    \int_{F_\eta}
    \pi(\theta\mid\bar y)\,d\theta.
\]

For a scalar function $g$ evaluated along a stationary chain, its
integrated autocorrelation time is denoted as $\tau_g$, which determines the effective sample size
approximately as
\(
    N_{\mathrm{eff}}
    \approx
    \frac{N}{\tau_g}.
\)
For implementations, we use an ensemble MCMC process with $N_w$ walkers and $N_s$ steps.

\section{Problem Formulation}
We aim to obtain $\pi(\theta\mid\bar y)$, the posterior distribution of the detailed model, in a computationally efficient manner. The key idea is to use the information from the learned posterior of a hierarchically related reduced model to speed up the inference of the higher-dimensional model (lower level in hierarchy). 
Note that the only ordering imposed on the hierarchical relations between models is that the higher level abstraction model in the pair chosen from the hierarchy has fewer states, that is,
\(
    \hat n<n.
\). In addition to this, we do not assume an ordering on $p$ and $\hat p$. This allows us to consider reduced models that contain lumped parameters and might lead to more parameters. Importantly, no mapping between $\theta$ and $\hat\theta$ is assumed. To formulate the problem of inference speedup, we define behavioral compatibility between two models in the hierarchy.
\begin{definition}[Behavioral compatibility]
For $\epsilon>0$, define
\[
    R_\epsilon
    :=
    \left\{
        (\theta,\hat\theta)\in\Theta\times\widehat\Theta:
        d\!\left(
            y(\cdot;\theta),
            \hat y(\cdot;\hat\theta)
        \right)
        \le
        \epsilon
    \right\}.
\]
For fixed $\theta$, define
\[
    R_\epsilon(\theta)
    :=
    \left\{
        \hat\theta\in\widehat\Theta:
        (\theta,\hat\theta)\in R_\epsilon
    \right\}.
\]
\end{definition}
The set $R_\epsilon(\theta)$ is the measurable set of reduced model parameters that are behaviorally compatible with the detailed model parameter $\theta$. Integrating the posterior mass over this set gives us a measure of the behavioral support.

We assume that, for every fixed $\theta\in\Theta$, the following map is Borel measurable:
\[
\hat\theta \mapsto
d\!\left(y(\cdot;\theta),\hat y(\cdot;\hat\theta)\right).
\]
\begin{definition}[Behavioral support]
The behavioral support assigned by the reduced model posterior to a detailed model parameter $\theta$ is
\[
    b_\epsilon(\theta)
    :=
    \int_{R_\epsilon(\theta)}
    \hat\pi(\hat\theta\mid\bar y)\,d\hat\theta.
\]
Equivalently, it is expressed as a probability measure
\[
    b_\epsilon(\theta)
    =
    \Pr_{\hat\theta\sim\hat\pi(\cdot\mid\bar y)}
    \left[
        d\!\left(
            y(\cdot;\theta),
            \hat y(\cdot;\hat\theta)
        \right)
        \le
        \epsilon
    \right].
\]
% \(
%     0\le b_\epsilon(\theta)\le 1.
% \)
\end{definition}
% \subsection{Behaviorally retained detailed-parameter set}
To formalize the notion of how much of the detailed model parameter space is retained based on behavioral compatibility, we introduce a behavioral support level $\gamma\in(0,1)$ and define the set
\[
    K_{\epsilon,\gamma}
    :=
    \left\{
        \theta\in\Theta:
        b_\epsilon(\theta)\ge\gamma
    \right\}.
\]
When $b_\epsilon$ is continuous, the boundary of $K_{\epsilon,\gamma}$, relative to $\Theta$, is contained in the level set
\(
    b_\epsilon(\theta)=\gamma,
\)
which we aim to learn for filtering detailed-model proposals.
To implement this idea, we propose a data-dependent behavioral prior.
\begin{proposition}
\label{prop:behavioral-prior}
Let $\pi_0(\theta)$ be the baseline prior. If
\[
    0<
    \int_\Theta
    \pi_0(\vartheta)b_\epsilon(\vartheta)\,d\vartheta
    <
    \infty,
\]
then, the following probability density function defines a proper data-dependent behavioral prior on $\Theta$:
\[
    \pi_B(\theta\mid\bar y)
    :=
    \frac{
        \pi_0(\theta)b_\epsilon(\theta)
    }{
        \int_\Theta
        \pi_0(\vartheta)b_\epsilon(\vartheta)\,d\vartheta
    }
\]
Thresholding the behavioral support instead gives the hard-restricted
posterior
\[
    \pi_{B,\gamma}(\theta\mid\bar y)
    \propto
    L(\bar y\mid\theta)\pi_0(\theta)
    \mathbf{1}_{K_{\epsilon,\gamma}}(\theta).
\]
This is an intentionally restricted posterior: its support is limited to
$K_{\epsilon,\gamma}$ rather than preserving the support of the baseline
posterior.
\end{proposition}

\begin{proof}
Both constructions are nonnegative. The stated normalization condition gives
a finite positive normalizing constant for $\pi_B$. The restricted posterior
is proper whenever
\(
0<\int_{K_{\epsilon,\gamma}}
L(\bar y\mid\vartheta)\pi_0(\vartheta)\,d\vartheta<\infty.
\)
\end{proof}

\section{Results}
% results_theory.tex
% Body-only theory section for insertion into the paper.
%
% Assumes the following have already been defined in the Preliminaries:
%   \mathcal H_T = L_2([0,T];\mathbb R^p)
%   d(y_1,y_2) = \|y_1-y_2\|_{\mathcal H_T}
%   approximate simulation relation S_\epsilon
%   finite-horizon Gramians Q_T, \hat Q_T, H_T
%   baseline and reduced Bayesian posteriors
%
% Suggested bibliography keys used below:
%   girard2007approximation
%   jones2004markov
%   qian2022balancing

\subsection{Nonlinear models}
Consider the two nonlinear system models in~\eqref{eq:nl-sys} - ~\eqref{eq:nl-sys-reduced} in an approximate simulation relation. Note that the initial conditions, $x(0)$, and $\hat x(0)$, are fixed and we assume the usual regularity conditions on $f$ and $\hat f$ required for existence and uniqueness of solutions on $[0,T]$ for all parameters. Using the approximate simulation relation, we propose how the behavioral compatibility tolerance $\epsilon$ can be chosen.
\begin{proposition}
\label{prop:approx-sim-certificate}
For every $\theta\in\Theta$, assume there exists a fixed $\hat\theta\in\widehat\Theta$ such that the prescribed initial pair $(x(0),\hat x(0))$ belongs to $S_\epsilon$, and transition matching holds for the corresponding parameterized dynamics throughout $[0,T]$.
For an approximate simulation relation between systems~\eqref{eq:nl-sys} and~\eqref{eq:nl-sys-reduced}, $S_\epsilon$, with observation precision $\epsilon/\sqrt{T}$, the choice of the behavioral compatibility tolerance $\epsilon$ is an upper bound on the directed Hausdorff distance between the detailed and reduced model behavior sets:
\[
    d_H(\mathcal Y,\widehat{\mathcal Y})
    \le
    \epsilon.
\]
\end{proposition}

\begin{proof}
Since $\|y(t)-\hat y(t)\|_2$ is bounded by the observation precision of the relation for all $t\in[0,T]$, we have 
\[
\begin{aligned}
    d^2(y,\hat y)
    =
    \int_0^T
    \|y(t)-\hat y(t)\|_2^2\,dt
    \le
    \int_0^T
    \frac{\epsilon^2}{T}\,dt
    =
    \epsilon^2.
\end{aligned}
\]
For each $\theta$, choosing the corresponding fixed $\hat\theta$ gives
$d_H(\mathcal Y,\widehat{\mathcal Y})\le\epsilon$.
\end{proof}
With a particular choice of $\epsilon$ satisfying Proposition~\ref{prop:approx-sim-certificate}, we can now prove the existence of a behavioral support for MCMC. 
\begin{theorem}
\label{thm:behavioral-support-existence}
If the higher-level model of reduced complexity admits samples
$\hat\theta^{(1)},\hat\theta^{(2)},\ldots$
from a Markov chain with invariant distribution
$\hat\pi(\cdot\mid\bar y)$ satisfying the Markov chain ergodic theorem for bounded measurable functions, then the behavioral support $b_{\epsilon}(\theta)$ is well defined and can be consistently estimated using the $\hat \theta$ samples as
\[
    \hat b_{\epsilon,M}(\theta)
    :=
    \frac{1}{M}
    \sum_{j=1}^{M}
    \mathbf 1
    \left[
        d\!\left(
            y(\cdot;\theta),
            \hat y(\cdot;\hat\theta^{(j)})
        \right)
        \le
        \epsilon
    \right],
\]
which, for every fixed $\theta$, gives 
\[
    \hat b_{\epsilon,M}(\theta)
    \xrightarrow[M\rightarrow\infty]{\mathrm{a.s.}}
    b_\epsilon(\theta).
\]
\end{theorem}

\begin{proof}
% Since the indicator of $R_\epsilon(\theta)$ is measurable, nonnegative, and bounded
% by one, its integral with respect to the proper reduced posterior exists and lies in $[0,1]$. 
For fixed $\theta$, the indicator used in
$\hat b_{\epsilon,M}(\theta)$ is a bounded measurable function of the
reduced model Markov chain. The Markov chain ergodic theorem therefore gives
\[
    \hat b_{\epsilon,M}(\theta)
    \longrightarrow
    \mathbb E_{\hat\pi}
    \left[
        \mathbf 1
        \left[
            d(y(\theta),\hat y(\hat\theta))
            \le\epsilon
        \right]
    \right]
    =
    b_\epsilon(\theta)
\]
almost surely; see, for example~\cite{jones2004markov}.
\end{proof}
% Theorem~\ref{thm:behavioral-support-existence} establishes that the behavioral support can be consistently estimated from samples of the higher level model posterior. 
Let $\Theta$ carry its Euclidean topology and suppose
$\theta\mapsto y(\cdot;\theta)$ is continuous from $\Theta$ into
$\mathcal H_T$. If the reduced posterior assigns zero probability to the
boundary
$d(y(\cdot;\theta),\hat y(\cdot;\hat\theta))=\epsilon$, then dominated
convergence implies that $b_\epsilon$ is continuous. Therefore,
$K_{\epsilon,\gamma}$ is closed, and $b_\epsilon(\theta)=\gamma$ provides a
well-posed level set learning problem. Before we can use this learned level
set to accelerate MCMC, we must show that this retained set contains the
desired data-fitting region of the detailed model.
\begin{theorem}
\label{thm:data-fit-retention}
For $\epsilon>\eta$ and $\gamma\in(0,1)$, if
\(
    \gamma
    \le
    m_{\hat\pi}(\epsilon-\eta),
\)
then, we have the inclusion
\(
    F_\eta
    \subseteq
    K_{\epsilon,\gamma}.
\)
\end{theorem}

\begin{proof}
For any $\theta\in F_\eta$ and any
$\hat\theta\in\widehat F_{\epsilon-\eta}$, the triangle inequality gives
\[
\begin{aligned}
    d(y,\hat y)
    \le
    d(y,\bar y)
    +
    d(\bar y,\hat y)
    \le
    \eta+(\epsilon-\eta)
    =
    \epsilon.
\end{aligned}
\]
Therefore,
\(
    \widehat F_{\epsilon-\eta}
    \subseteq
    R_\epsilon(\theta).
\)
and it follows that
\[
\begin{aligned}
    b_\epsilon(\theta)
    &=
    \int_{R_\epsilon(\theta)}
    \hat\pi(\hat\theta\mid\bar y)\,d\hat\theta
    \ge
    \int_{\widehat F_{\epsilon-\eta}}
    \hat\pi(\hat\theta\mid\bar y)\,d\hat\theta\\
    &=
    m_{\hat\pi}(\epsilon-\eta)
    \ge
    \gamma.
\end{aligned}
\]
Hence $\theta\in K_{\epsilon,\gamma}$ and since $\theta$ was arbitrary, 
$F_\eta\subseteq K_{\epsilon,\gamma}$.
\end{proof}
This is the main result needed to guarantee that the  behavioral restriction retains every detailed parameter whose simulated output satisfies the chosen data fit tolerance, provided that the reduced posterior assigns
sufficient mass to trajectories within $\epsilon-\eta$ of the data.
\subsection{Linear models with parametric initial conditions}
For the special case when parameters only appear in the initial conditions, we can derive closed-form expressions for the optimal reduced representation and the exact behavioral support. Consider a linear autonomous system, $\dot x =Ax$, with $x(0)=\theta\in\Theta\subseteq\mathbb R^n$,
and the reduced complexity linear system
\(
    \dot{\hat x} =\hat A\hat x
    \) 
    with $\hat x(0)=\hat\theta\in\widehat\Theta
    \subseteq\mathbb R^{\hat n}$, the finite horizon outputs are
\(
    y(t;\theta)=Ce^{At}\theta,
    \hat y(t;\hat\theta)
    =
    \hat C e^{\hat A t}\hat\theta.
\). 
The behavioral compatibility condition imposed in this case gives rise to a finite-horizon least-squares projection, which provides a closed-form mapping from the detailed initial condition to the optimal reduced initial condition in this special case. 
\begin{theorem}
\label{thm:linear-optimal-representation}
For every $\theta\in\mathbb R^n$, the unique reduced initial condition that minimizes $d^2(y(\cdot;\theta),\hat y(\cdot;\hat\theta))$ is
\(
    \hat\theta^\star(\theta)
    =
    \hat Q_T^{-1}H_T^\top\theta,
\)
and the minimum value is
\[
    \inf_{\hat\theta\in\mathbb R^{\hat n}}
    d^2\!\left(
        y(\cdot;\theta),
        \hat y(\cdot;\hat\theta)
    \right)
    =
    \theta^\top E_T\theta
\]
where 
\(
    E_T\) is a positive semidefinite matrix defined as
\(
    E_T :=
    Q_T-H_T\hat Q_T^{-1}H_T^\top.
\)
% Moreover,
% \[
% \begin{aligned}
%     d^2(y,\hat y)
%     &=
%     \theta^\top E_T\theta\\
%     &\quad+
%     \left(
%         \hat\theta-\hat\theta^\star(\theta)
%     \right)^\top
%     \hat Q_T
%     \left(
%         \hat\theta-\hat\theta^\star(\theta)
%     \right).
% \end{aligned}
% \]
\end{theorem}

\begin{proof}
Expanding the squared finite-horizon distance gives
\[
    d^2(y,\hat y)
    =
    \theta^\top Q_T\theta
    -
    2\theta^\top H_T\hat\theta
    +
    \hat\theta^\top\hat Q_T\hat\theta.
\]
Then, we can complete the square in $\hat\theta$ to obtain
\begin{equation}
\begin{aligned}
    d^2(y,\hat y)
    &=
    \theta^\top
    \left(
        Q_T-H_T\hat Q_T^{-1}H_T^\top
    \right)
    \theta\\
    &\quad+
    \left(
        \hat\theta-\hat Q_T^{-1}H_T^\top\theta
    \right)^\top
    \hat Q_T
    \left(
        \hat\theta-\hat Q_T^{-1}H_T^\top\theta
    \right).
\end{aligned}
\label{eq:linear-completed-square}
\end{equation}
Since $\hat Q_T\succ0$, the second term is nonnegative and vanishes uniquely
at
\(
    \hat\theta
    =
    \hat Q_T^{-1}H_T^\top\theta.
\)
The minimum value is therefore $\theta^\top E_T\theta$ and its nonnegativity implies $E_T\succeq0$.
\end{proof}
The closed-form of the directed Hausdorff distance immediately follows from this Theorem. 
\begin{corollary}
\label{cor:linear-exact-hausdorff}
Suppose the reduced parameter domain contains
$\hat\theta^\star(\theta)$ for every $\theta\in\Theta$. Then
\(
    d_H(\mathcal Y,\widehat{\mathcal Y})
    =
    \sup_{\theta\in\Theta}
    \sqrt{\theta^\top E_T\theta}
\)
If $\norm{\theta}_2\le R$ for all $\theta\in\Theta$, then
\(
    d_H(\mathcal Y,\widehat{\mathcal Y})
    \le
    R\sqrt{\lambda_{\max}(E_T)}.
\)
Equality holds when $\Theta$ is the Euclidean ball
$\{\theta:\norm{\theta}_2\le R\}$.
Therefore, the choice of $\epsilon$ can be certified as
\(
    \epsilon
    \ge
    R\sqrt{\lambda_{\max}(E_T)}.
\)
\end{corollary}
Additionally, from~\eqref{eq:linear-completed-square}, we also have a closed-form expression for the behavioral compatibility set.
\begin{corollary}
\label{cor:linear-compatibility-set}
Under the assumption of Corollary~\ref{cor:linear-exact-hausdorff},
for a fixed detailed initial condition $\theta$, behavioral compatibility is
possible if and only if
\(
    \epsilon^2-\theta^\top E_T\theta
    \ge0.
\)
Then, the behavioral compatibility set is
\[
\begin{aligned}
    R_\epsilon(\theta)
    =
    \Big\{
        \hat\theta\in\widehat\Theta:
        &
        \left(
            \hat\theta-\hat\theta^\star(\theta)
        \right)^\top
        \hat Q_T
        \left(
            \hat\theta-\hat\theta^\star(\theta)
        \right)\\
        &
        \le
        \epsilon^2-\theta^\top E_T\theta
    \Big\}.
\end{aligned}
\]
Thus $R_\epsilon(\theta)$ is the intersection of the reduced parameter domain
with an ellipsoid centered at
$\hat\theta^\star(\theta)$.
\end{corollary}
% \begin{proof}
% The result follows directly from the completed-square expression in
% Theorem~\ref{thm:linear-optimal-representation}.
% \end{proof}

\subsubsection{Behavioral support in the specialized linear case}

For the linear initial-condition model,
\[
\begin{aligned}
    b_\epsilon(\theta)
    &=
    \int_{R_\epsilon(\theta)}
    \hat\pi(\hat\theta\mid\bar y)\,d\hat\theta\\
    &=
    \Pr_{\hat\theta\sim\hat\pi(\cdot\mid\bar y)}
    \left[
        \left(
            \hat\theta-\hat\theta^\star(\theta)
        \right)^\top
        \hat Q_T
        \left(
            \hat\theta-\hat\theta^\star(\theta)
        \right)
        \le
        \epsilon^2-\theta^\top E_T\theta
    \right]
\end{aligned}
\]
whenever
$\epsilon^2-\theta^\top E_T\theta\ge0$, and
\(
    b_\epsilon(\theta)=0
\)
otherwise.

For a reduced posterior family for which ellipsoidal probabilities can be
evaluated analytically, $b_\epsilon(\theta)$ and the level set
\(
    b_\epsilon(\theta)=\gamma
\)
can therefore be computed directly.
% \paragraph{Outcome.}
Therefore, the specialized linear case provides a closed-form behavioral support formulation, which we use to validate (and visualize) the GP level-set learning procedure.
% \subsubsection{Specialized retention statement}

% Combining Corollary~\ref{cor:linear-exact-hausdorff} with
% Theorem~\ref{thm:data-fit-retention} gives the complete theoretical chain for
% the initial-condition case:
% \[
%     d_H(\mathcal Y,\widehat{\mathcal Y})
%     \;\longrightarrow\;
%     \epsilon
%     \;\longrightarrow\;
%     b_\epsilon(\theta)
%     \;\longrightarrow\;
%     K_{\epsilon,\gamma}
%     \;\longrightarrow\;
%     F_\eta\subseteq K_{\epsilon,\gamma}.
% \]

% For the ball domain $\Theta_R$, an exact compatibility certificate is
% available from
% \[
%     \epsilon
%     \ge
%     R\sqrt{\lambda_{\max}(E_T)}.
% \]
% For any such $\epsilon>\eta$, if
% \[
%     \gamma
%     \le
%     m_{\hat\pi}(\epsilon-\eta),
% \]
% then every detailed initial condition satisfying
% \[
%     d(y(\theta),\bar y)\le\eta
% \]
% is retained by the behavioral-support restriction.

% \paragraph{Outcome.}
% The parametric initial-condition case closes the full theory analytically:
% it gives the best reduced representation, the exact model-level mismatch, a
% principled compatibility tolerance, an explicit compatibility region, and a
% retention guarantee for detailed-model data fit.

\subsection{Computational Implementation}
\subsubsection{Efficiency-exploration tradeoff}
The results above can be applied to speed up MCMC-based parameter inference of a lower-level, detailed model in a hierarchy. The level of efficiency vs. exploration tradeoff can be tuned since the formulations above are monotonic. That is, for every fixed $\theta$,
\(
    \epsilon_1\le\epsilon_2\) implies \(
    b_{\epsilon_1}(\theta)
    \le
    b_{\epsilon_2}(\theta).
\)
So, for fixed $\gamma$,\( 
    K_{\epsilon_1,\gamma}
    \subseteq
    K_{\epsilon_2,\gamma}.
\). Finally, for a fixed $\epsilon$, we have from the definition of a superlevel set in $\gamma$ that $\gamma_1 \le \gamma_2$ implies \(
    K_{\epsilon,\gamma_2}
    \subseteq
    K_{\epsilon,\gamma_1}.
\)
The monotonicity arguments hold true since for every $\epsilon_1\le\epsilon_2$, 
\(
    R_{\epsilon_1}(\theta)
    \subseteq
    R_{\epsilon_2}(\theta).
\)
integration over nested sets gives
\(
    b_{\epsilon_1}(\theta)
    \le
    b_{\epsilon_2}(\theta).
\)
The tradeoff is that for higher $\epsilon$, the retained detailed parameter region is bigger, which means less efficient but more exploratory inference. Similarly, increasing $\gamma$ makes the retained set smaller. These thresholds are constrained by Theorem~\ref{thm:data-fit-retention}, which can be used when 
$F_\eta$ is the desired data fit level.

\subsubsection{Learning the behavioral compatibility level set using GPs}
% Let
% \[
%     \hat\theta^{(1)},\ldots,\hat\theta^{(M)}
% \]
% denote samples from a stationary ergodic reduced model MCMC run.
% For fixed $\theta$, define
% \[
%     I_j(\theta)
%     :=
%     \mathbf{1}
%     \left[
%         d\!\left(
%             y(\cdot;\theta),
%             \hat y(\cdot;\hat\theta^{(j)})
%         \right)
%         \le
%         \epsilon
%     \right].
% \]
% The Monte Carlo estimator of behavioral support is
% \[
%     \hat b_{\epsilon,M}(\theta)
%     :=
%     \frac{1}{M}
%     \sum_{j=1}^{M}
%     I_j(\theta).
% \]

% Under the usual ergodic conditions,
% \[
%     \hat b_{\epsilon,M}(\theta)
%     \longrightarrow
%     b_\epsilon(\theta)
% \]
% as $M\to\infty$.
% If the indicator process has summable autocovariances, its asymptotic variance
% is
% \[
%     \operatorname{Var}
%     \left[
%         \hat b_{\epsilon,M}(\theta)
%     \right]
%     \approx
%     \frac{\sigma_I^2(\theta)}{M},
% \]
% where
% \[
%     \sigma_I^2(\theta)
%     :=
%     \operatorname{Var}[I_0(\theta)]
%     +
%     2\sum_{k=1}^{\infty}
%     \operatorname{Cov}
%     \left(
%         I_0(\theta),
%         I_k(\theta)
%     \right).
% \]
% Equivalently, when an integrated autocorrelation time $\tau_I(\theta)$ is
% well defined,
% \[
%     \operatorname{Var}
%     \left[
%         \hat b_{\epsilon,M}(\theta)
%     \right]
%     \approx
%     \frac{
%         b_\epsilon(\theta)
%         \left(1-b_\epsilon(\theta)\right)
%         \tau_I(\theta)
%     }{M}.
% \]
% This variance can be used as the observation-noise scale supplied to the GP.

% \subsubsection{GP level-set approximation}

For fixed $\epsilon$ and $\gamma$, we define a validity function as the target for a GP level-set approximation:
\(
    v_{\epsilon,\gamma}(\theta)
    :=
    b_\epsilon(\theta)-\gamma.
\)
Then, we can write the retained set as a superlevel set of this function,
\(
    K_{\epsilon,\gamma}
    =
    \{\theta:v_{\epsilon,\gamma}(\theta)\ge0\}.
\)
The GP training value at an evaluated detailed parameter is
\(
    \hat v_{\epsilon,\gamma,M}(\theta)
    :=
    \hat b_{\epsilon,M}(\theta)-\gamma.
\)  
Let $\mathcal D_n$ denote the resulting set of $n$ detailed parameter evaluations and empirical validity values. After $n$ detailed evaluations, we have the GP posterior predictive distribution as
\[
    v_{\epsilon,\gamma}(\theta)
    \mid
    \mathcal D_n
    \sim
    \mathcal N
    \left(
        \mu_n(\theta),
        \sigma_n^2(\theta)
    \right).
\]
So, the posterior probability that a detailed model parameter belongs to the retained
set is
\[
    p_{K,n}(\theta)
    :=
    \Pr
    \left[
        v_{\epsilon,\gamma}(\theta)\ge0
        \mid
        \mathcal D_n
    \right]
    =
    \Phi
    \left(
        \frac{\mu_n(\theta)}{\sigma_n(\theta)}
    \right).
\]
Once we have this GP, we can use it to reject parameter proposals when the posterior probability of being in the retained set is below a threshold $\zeta$,
\(
    p_{K,n}(\theta)
    \le
    \zeta.
\)
The trained GP is then used as a behavioral prior in the detailed model MCMC inference.

\subsubsection{Threshold selection}
The compatibility tolerance $\epsilon$ specifies the maximum accepted output
discrepancy, while $\eta<\epsilon$ defines the desired detailed model data-fit
tolerance in the retention theorem. The support level $\gamma$ is the minimum
reduced posterior compatibility mass required to retain a detailed parameter.
The GP filtering level $\zeta$ is the minimum posterior probability of set
membership required to simulate a proposal. Increasing either $\gamma$ or
$\zeta$ makes the behavioral restriction more selective.
Algorithm~\ref{alg:behavioral-guidance} summarizes the procedure.

\begin{algorithm}[t]
\caption{Behaviorally guided detailed model inference}
\label{alg:behavioral-guidance}
\begin{algorithmic}[1]
\REQUIRE Reduced posterior samples
$\{\hat\theta^{(j)}\}_{j=1}^{M}$ and precomputed reduced trajectories,
compatibility tolerance $\epsilon$, support level $\gamma$, GP filtering 
level $\zeta$, initial detailed design, GP-MPM budget, detailed model prior and
likelihood.
\STATE Evaluate $\hat b_{\epsilon,M}(\theta)$ on the initial detailed design.
\STATE Fit a GP to
$\hat v_{\epsilon,\gamma,M}(\theta)
=
\hat b_{\epsilon,M}(\theta)-\gamma$.
\FOR{each GP-MPM iteration}
    \STATE Select the next detailed parameter using the GP-MPM acquisition rule
    \cite{moorse2024validation}.
    \STATE Simulate the detailed model once at the selected parameter.
    \STATE Compute $\hat b_{\epsilon,M}$ by comparison with the stored reduced
    posterior trajectories.
    \STATE Update the GP training set and refit the GP.
\ENDFOR
\STATE Store the trained GP for detailed-model inference.
\STATE Initialize the detailed model MCMC ensemble.
\FOR{each detailed model MCMC proposal $\theta^\star$}
    \STATE Evaluate
    $
        p_K(\theta^\star)
        =
        \Phi(
        \mu(\theta^\star)/\sigma(\theta^\star)
        )
    $.
    \IF{$p_K(\theta^\star)\le\zeta$}
        \STATE Reject $\theta^\star$ without a detailed model simulation.
    \ELSE
        \STATE Simulate the detailed model at $\theta^\star$.
        \STATE Evaluate the detailed likelihood and apply the ordinary MCMC
        acceptance rule.
    \ENDIF
\ENDFOR
\RETURN Behaviorally filtered detailed model posterior samples.
\end{algorithmic}
\end{algorithm}

\subsubsection{Computational efficiency analysis}
Let $n$ be the number of detailed model simulations used to train the GP, let
$N$ be the number of subsequent detailed model MCMC proposals that pass the ordinary parameter priors, and let $p_K$
be the probability that such a proposal passes the GP filter. Denote the costs of a
reduced model inference, a detailed model simulation, and a GP-filter
evaluation by $c_R$, $c_F$, and $c_G$, respectively. Let $c_{\mathrm{GP}}$
denote the total cost of GP fitting and GP-MPM acquisition during training.
Ignoring the comparatively
small cost of comparing one detailed trajectory with the stored reduced
trajectories, the detailed model inference cost is
\(
    C_{\mathrm{base}}
    =
    Nc_F.
\)
The behaviorally guided cost is approximately
\[
    C_{\mathrm{guided}}
    =
    c_R
    +
    nc_F
    +
    c_{\mathrm{GP}}
    +
    Nc_G
    +
    Np_Kc_F.
\]
The GP-guided inference is computationally favorable when
\[
    N(1-p_K)c_F
    >
    c_R+nc_F+c_{\mathrm{GP}}+Nc_G.
\]
The expected number of detailed simulations avoided during the MCMC stage is
$N(1-p_K)$. We report computational improvements as percentage reductions in
MCMC elapsed time relative to baseline. Let $T_{\mathrm{base}}$ and
$T_{\mathrm{guided}}$ denote the measured baseline and guided MCMC times.
The reported improvement is
\[
    \text{MCMC time reduction}
    =100\left(1-\frac{T_{\mathrm{guided}}}{T_{\mathrm{base}}}\right)\%.
\]
Reduced-model inference and GP training are one-time setup costs,
reported separately in the gene-expression application.

\section{Applications}
\subsection{Illustrative linear system example}
\begin{figure}[htbp]
    \centering
    \includegraphics[width=0.98\textwidth]{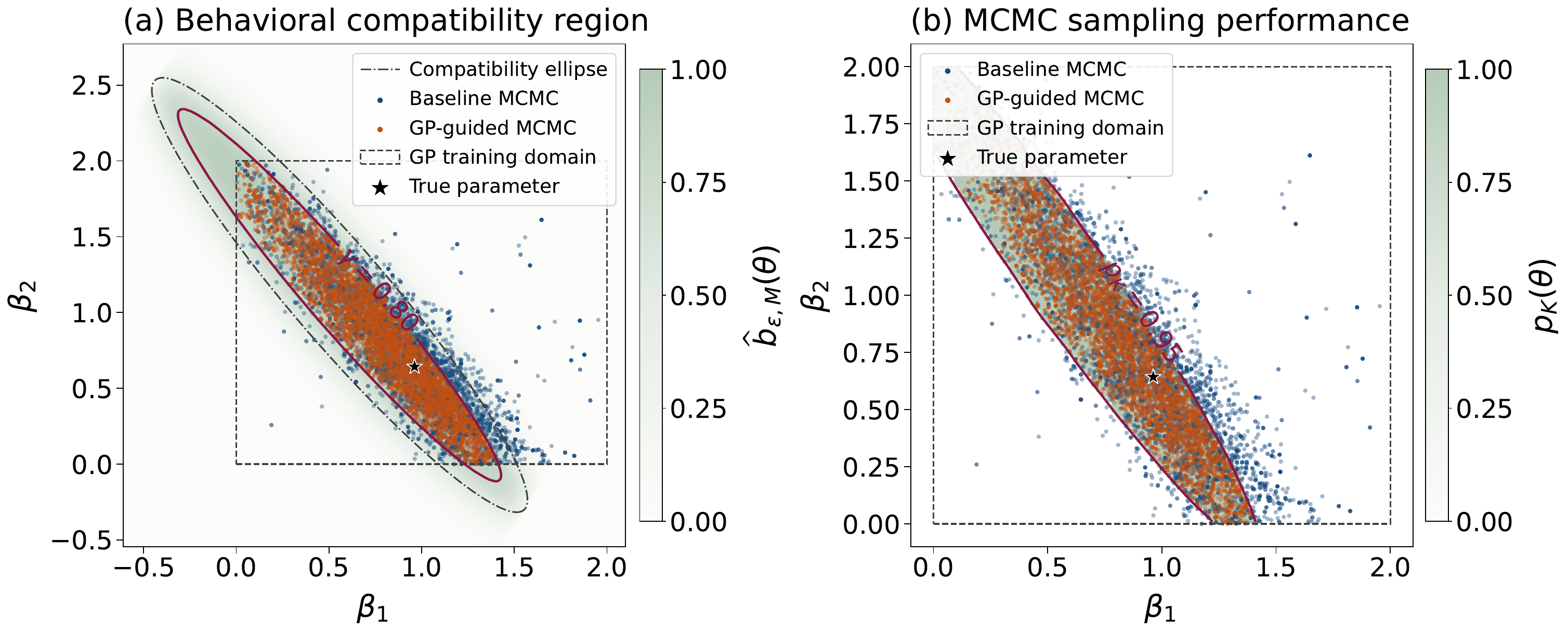}
    \caption{\footnotesize (a) Behavioral compatibility region and its boundary ($\hat b_{\epsilon,M}(\theta)=\gamma=0.80$). The dash-dotted
    compatibility ellipse is computed at the fixed reduced parameter
    $\alpha\approx1.627$, which is the mean of the reduced posterior. (b) MCMC sampling performance of the baseline inference (blue) and the GP-guided inference samples (orange). We observe that the baseline MCMC explores outside the learned acceptance region of the GP. The $\star$ marks the ground truth.}
    \label{fig:joint-parameter-results}
\end{figure}
Consider a two-state detailed model with
$A=\operatorname{diag}(-a_1,-a_2)$, $a_1=1$, $a_2=2$, and output
$y(t;\beta)=\beta_1e^{-a_1t}+\beta_2e^{-a_2t}$.
The reduced model has one state with matrix $\hat A=-a_r$, $a_r=1.5$, and output
$\hat y(t;\alpha)=\alpha e^{-a_rt}$.
For this example, we consider that the parameters are the initial conditions,
$x_0=\theta=(\beta_1,\beta_2)^{\mathsf T}$ and
$\hat x_0=\hat\theta=\alpha$.
We construct synthetic observations at $1001$ equally spaced times on $[0,10]$ by adding independent zero-mean Gaussian noise with standard deviation $\sigma_{\mathrm{obs}}=0.02$ to the detailed trajectory.
We set the ground truth 
$\theta^{\mathrm{true}}\approx(0.96,0.64)^{\mathsf T}$, a point chosen at random.
For both model outputs, $g$, we define the log-likelihood formulation for this illustrative example as
\[
    \log L(\vartheta)
    =-\frac{d_T^2(g(\vartheta),\bar y)}{2s_L^2}+\mathrm{const},
    \qquad
    d_T^2=\int_0^T\bigl(g(t;\vartheta)-\bar y(t)\bigr)^2\,dt.
\]
Here $s_L$ is an independent likelihood scale, which we set to $0.1$. The priors are
$\beta_1,\beta_2\sim\operatorname{Uniform}(0,2)$ for both detailed fits and
$\alpha\sim\mathcal N(1,0.5^2)$. For the baseline, reduced, and GP-guided inference, we use $16$ walkers and $1000$
steps, discarding the first $200$ steps for burn-in. To train the GP, we use $6$ initial
evaluations followed by $100$ GP-MPM~\cite{moorse2024validation} iterations on the same coordinate space $[0,2]^2$.

We note that at the ground truth parameter values, Theorem~\ref{thm:linear-optimal-representation}
gives $\alpha^*=1.705$ and a minimum integrated squared error of
$0.0088$, corresponding to an $L^2$ distance of $0.094$. Now, we use our results above to set the thresholds. Since we have added an independent Gaussian observation noise, we set
\[
    \eta=\sigma_{\mathrm{obs}}\sqrt{w_{\max}\chi^2_{N,0.95}},
    \qquad
    \epsilon=\eta+q_\gamma,
\]
where $\chi^2_{N,0.95}$ is the $0.95$-quantile of a chi-squared variable, $w_{\max}$ is a weight, $q_\gamma$ is the $\gamma$-quantile of reduced posterior data
discrepancy, $d_T(\hat y(\alpha),\bar y)$.
% Gaussian conjugacy gives the reduced posterior mean
% $1.6268$ and standard deviation $0.1637$, so its discrepancy quantile
% can be computed directly. 
So, with a choice of $\gamma=0.80$ and GP filtering level
$\zeta=0.95$, we obtain $\eta\approx0.0656$, $q_\gamma\approx0.1747$,
and $\epsilon\approx0.2403$. By construction, we note that
$m_{\hat\pi}(\epsilon-\eta)=\gamma$, which satisfies
Theorem~\ref{thm:data-fit-retention}. The empirical estimate of the behavioral compatibility region,
$\{\theta:\hat b_{\epsilon,M}(\theta)\ge\gamma\}$ is shown in Figure~\ref{fig:joint-parameter-results}a.
At a fixed reduced parameter $\alpha=\mu_\alpha$,
where $\mu_\alpha$ is the reduced posterior mean, compatibility in
the detailed parameters is the ellipse
$\beta^{\mathsf T}Q_T\beta-2\mu_\alpha H_T^{\mathsf T}\beta
+\hat Q_T\mu_\alpha^2\le\epsilon^2$, shown separately in the figure.
The trained GP filtering probability and the comparison of the baseline
and GP-guided MCMC inference over $[0,2]^2$ prior domain is shown in Figure~\ref{fig:joint-parameter-results}b.
\FloatBarrier

\subsection{Inference of gene-expression parameters}
\label{sec:example-two}
To demonstrate the acceleration of Bayesian inference, we show an application where we infer the kinetic parameters of deGFP expression in a PURE cell-free system. For this application, we use previously published experimental measurements~\cite{jurado2026nucleotide} of deGFP fluorescence for different initial DNA concentrations. We show how our framework can be applied to a hierarchy of models. We build three levels of gene expression models, each with different levels of complexity. The high-level model is a two-state model that describes mRNA and GFP dynamics. The phenomenological model is a nine-state model that represents transcription, translation, and RNA degradation through binding complexes. The mechanistic model is a 38-state model that represents closed and active complexes with up to five polymerases per DNA and ten ribosomes per mRNA. We use the high-level model to learn the behavioral compatibility of the phenomenological and mechanistic models. We then use the learned behavioral compatibility to accelerate MCMC inference of the phenomenological and mechanistic models.

\subsubsection{Experimental data}
From~\cite{jurado2026nucleotide}, we have three replicate time-series trajectories of deGFP at six initial DNA
concentrations: $1.99$, $1.06$, $0.47$, $0.26$, $0.12$ and
$0.07\,\mathrm{nM}$. We post-process the data to retain measurements between $0.25$ and $3$ hours to remove the initial lag in the measurements and the stationary protein expression since we do not capture these effects in our models, and usually, the rate of expression after lag-phase and before stationary-phase is of interest. We convert the fluorescence arbitrary units to deGFP concentration using $5542.4$ fluorescence units per $\mu\mathrm{M}$ (from~\cite{jurado2026nucleotide}). 

\subsubsection{Hierarchical modeling}
\begin{figure*}[!htb]
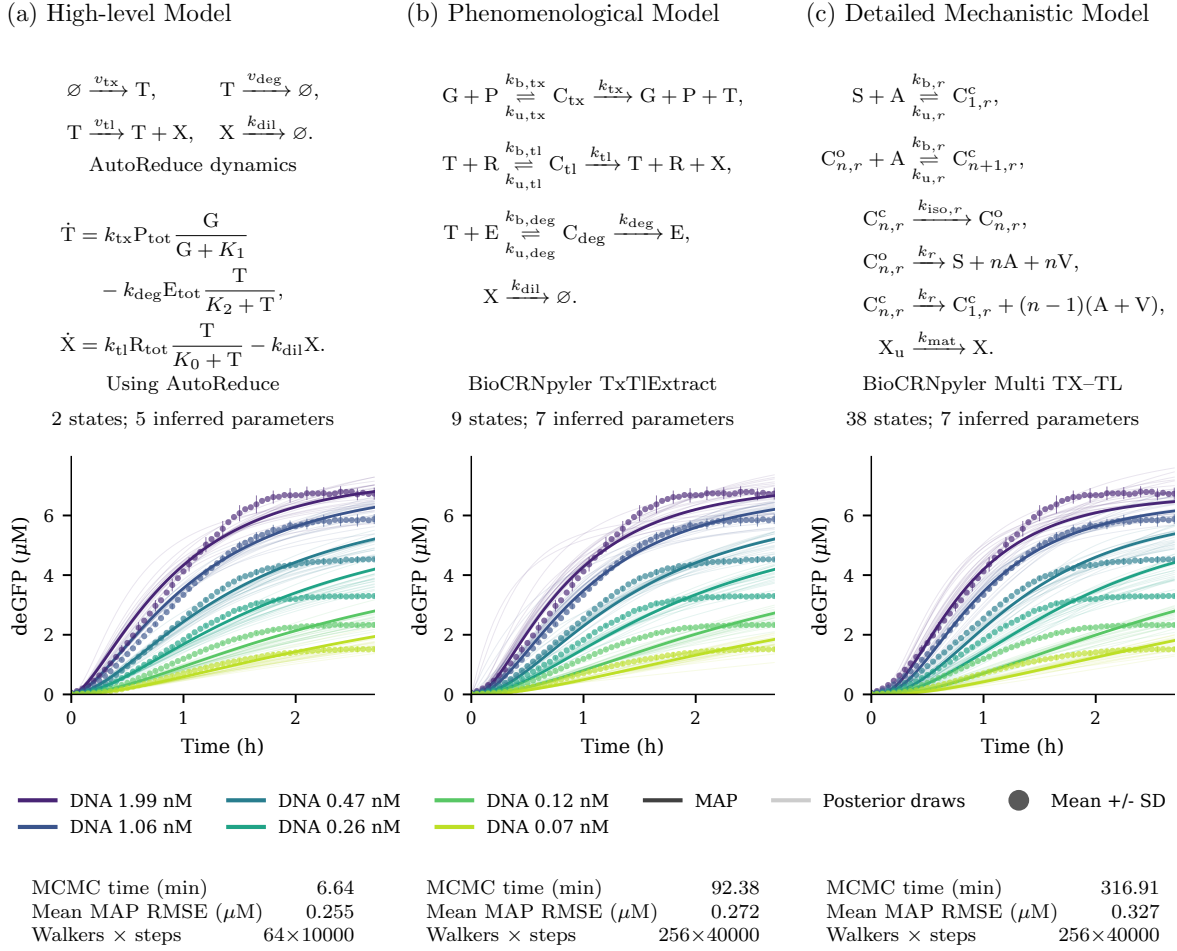

\centering
\setlength{\tabcolsep}{0pt}
\resizebox{0.94\textwidth}{!}{%
\begin{tabular}{@{}p{0.318\textwidth}@{\hspace{0.023\textwidth}}p{0.318\textwidth}@{\hspace{0.023\textwidth}}p{0.318\textwidth}@{}}
\raggedright\normalfont (a) High-level Model &
\raggedright\normalfont (b) Phenomenological Model &
\raggedright\arraybackslash\normalfont (c) Detailed Mechanistic Model\\[4pt]
\GeneHierarchyPanel{\GeneHighLevelCRN} &
\GeneHierarchyPanel{\GenePhenomenologicalCRN} &
\GeneHierarchyPanel{\GeneDetailedCRN}\\[3pt]
\centering\footnotesize Using AutoReduce &
\centering\footnotesize BioCRNpyler TxTlExtract &
\centering\arraybackslash\footnotesize BioCRNpyler Multi TX--TL\\[2pt]
\centering\footnotesize 2 states; 5 inferred parameters &
\centering\footnotesize 9 states; 7 inferred parameters &
\centering\arraybackslash\footnotesize 38 states; 7 inferred parameters\\[7pt]
\GeneHighLevelFit & \GenePhenomenologicalFit & \GeneDetailedFit\\[3pt]
\multicolumn{3}{@{}p{\textwidth}@{}}{\GeneHierarchyLegend}\\[7pt]
\centering\footnotesize\GeneHighLevelMetrics &
\centering\footnotesize\GenePhenomenologicalMetrics &
\centering\arraybackslash\footnotesize\GeneDetailedMetrics
\end{tabular}%
}
\caption{Hierarchical modeling of gene expression starting from a high-level mRNA and protein model, to a phenomenological model, and finally to a detailed mechanistic model. (top) CRNs or the ODE representation of each model. (middle) Model predictions (using samples from baseline posterior distributions) alongside experimental data. (bottom) Baseline inference times, steps/walkers, and mean
condition-wise MAP RMSEs.}
\label{fig:ex2-hierarchy}
\end{figure*}

The high-level model describes mRNA and GFP dynamics with two states. We use AutoReduce~\cite{pandey2023robustness} to obtain these
dynamics from conservation laws, quasi-steady-state elimination of binding complexes and parameter lumping. These are commonly used approaches to obtain reduced-order models~\cite{pandey2020modelreduction}. The next level model is a phenomenological model. We use BioCRNpyler~\cite{poole2022biocrnpyler}, a biological CRN compiler tool, to create this model. In BioCRNpyler, we specify the \texttt{TxTlExtract} mixture and mechanisms for 
transcription, translation and RNA degradation that generate chemical reactions with complexes (see Figure~\ref{fig:ex2-hierarchy}). This phenomenological model has 9 states. Finally, the most detailed level model in our application is a model of gene expression that also includes multiple polymerase and ribosome binding complexes. We use BioCRNpyler to create this model as well. This mechanistic model uses the multi TX-TL mechanisms, which represent closed and active complexes with up to five polymerases per DNA and ten ribosomes per mRNA. The model also includes a GFP maturation step. In total, this mechanistic model contains 38 states.
We show the CRNs of these models in Figure~\ref{fig:ex2-hierarchy}, with $\text{G},\text{P},\text{T},\text{R},\text{X},\text{E}$ denoting the DNA, RNA polymerase, mRNA, ribosome, fluorescent GFP and RNase respectively. The complexes $\text{C}_{\text{tx}},\text{C}_{\text{tl}},\text{C}_{\text{deg}}$ represent their binding complexes. 

For the detailed mechanistic model, we write the transcription and translation reactions in a common block form. For transcription, $r=\text{tx}$ and $(\text{S},\text{A},\text{V},N)=(\text{G},\text{P},\text{T},5)$; for translation, $r=\text{tl}$ and $(\text{S},\text{A},\text{V},N)=(\text{T},\text{R},\text{X}_{\text{u}},10)$. The complex $\text{C}^{\text{o}}_{n,r}$ has $n$ active enzymes, and $\text{C}^{\text{c}}_{n,r}$ has $n-1$ active enzymes and one closed enzyme. Binding uses $1\le n<N$, activation and open release use $1\le n\le N$, and closed release uses $2\le n\le N$. Here, $\text{X}_{\text{u}}$ is immature GFP, which matures into fluorescent GFP, $\text{X}$. The mechanistic model uses the RNase-mediated mRNA degradation and GFP loss reactions of the phenomenological model. GFP loss also acts on immature GFP with the same rate $k_{dil}$, as before.

The model hierarchy has different levels of complexity, each with its own set of parameters. The two-state model has five inferred parameters,
\(
    \hat\theta=\begin{bmatrix}k_{tx}&K_1&k_{tl}&k_{dil}&K_0\end{bmatrix}^{\top}.
\)
The phenomenological model has 10 kinetic parameters, and the mechanistic model has 13 kinetic parameters. The gene expression data is from a PURE cell-free protein expression system, which has minimal nuclease activity~\cite{neb_purexpress}. So, we set the total concentration for the pool of the RNase enzyme to a small value of $0.001\,\mu\mathrm{M}$. Further, we also fix the RNase binding, unbinding and catalytic mRNA degradation parameters to $k_{b,deg}=10\,\mu\mathrm{M}^{-1}\mathrm{min}^{-1}$, $k_{u,deg}=1\,\mathrm{min}^{-1}$ and $k_{deg}=0.1\,\mathrm{min}^{-1}$, respectively.
%  the lumped binding constants are
% \[
%     K_1=\frac{k_{u,tx}+k_{tx}}{k_{b,tx}},\qquad
%     K_0=\frac{k_{u,tl}+k_{tl}}{k_{b,tl}},\qquad
%     K_2=\frac{k_{u,deg}+k_{deg}}{k_{b,deg}}.
% \]
For the high-level model, we also fix $K_2=0.11\,\mu\mathrm{M}$ and $k_{deg}=0.1\,\mathrm{min}^{-1}$.
Consequently, both the phenomenological and the mechanistic models have the same set of parameters that are inferred from the data: 
\(
    \theta=\begin{bmatrix}
        k_{b,tx}&k_{u,tx}&k_{b,tl}&k_{u,tl}&k_{tx}&k_{tl}&k_{dil}
    \end{bmatrix}^{\top}.
\)

We also note that we fix $k_{iso,tx}=k_{iso,tl}=10\,\mathrm{min}^{-1}$ and $k_{mat}=1/8.25\,\mathrm{min}^{-1}$ in the detailed mechanistic model. For all models, the polymerase, ribosome and RNase resource pools are $0.1$, $1.2$ and $0.001\,\mu\mathrm{M}$, respectively. Once the hierarchy of models is setup, we use the high-level model of mRNA and protein to learn the behavioral compatibility with the phenomenological and mechanistic models. We denote these learned mappings by $\mathrm{GP}_{\text{hp}}$ and
$\mathrm{GP}_{\text{hm}}$, respectively, for ``high-phenomenological'' and ``high-mechanistic''.

\subsubsection{Baseline Bayesian inference}
We use Bioscrape~\cite{pandey2023bioscrape} for parameter inference, which provides a wrapper around the \texttt{emcee} ensemble sampler~\cite{foremanmackey2013emcee} suitable for biological system models.
We use broad log-uniform priors on the kinetic parameters with bounds chosen according to the biological meaning of each parameter. We have $k_{b,tx}\in[10^{-1},10^{5}], k_{u,tx}\in[10^{-3},10^{3}]$, $k_{b,tl}\in[10^{-2},10^{4}], k_{u,tl}\in[10^{-3},10^{3}]$, $k_{tx}\in[10^{-2},10^{4}], k_{tl}\in[10^{-3},10^{2}],k_{dil}\in[10^{-5},10]$.
For the two-state model, the priors are also log-uniform with different ranges for each parameter chosen according to biological plausibility. We have
$k_{tx}\in[10^{-1},10^3]$, $k_{tl}\in[10^{-2},10]$,
$k_{dil}\in[10^{-4},1]$ and $K_0,K_1\in[10^{-3},10]\,\mu\mathrm{M}$.
For baseline inference, we use $\GeneDetailedWalkers$ walkers, $\GeneDetailedSteps$ steps with $\GeneDetailedBurn$ as burn-in steps. To initialize the walkers, we sample from a Gaussian ball with $5\%$ relative standard
deviation around the nominal model rates. For the inference of the two-state model, we use $\GeneReducedWalkers$ walkers, $\GeneReducedSteps$ steps with $5000$
burn-in steps. We compute the sampled maximum a posteriori (MAP) estimate as the joint parameter tuple with the largest post-burn log posterior. The baseline fits in Figure~\ref{fig:ex2-hierarchy} show experimental means
and standard deviations, sampled MAP trajectories and 20 joint posterior
draws at every DNA concentration. Mean condition-wise MAP RMSEs are
$\GeneReducedRMSE$, $\GenePhenomenologicalRMSE$ and
$\GeneBaselineRMSE\,\mu\mathrm{M}$ for the high-level, phenomenological and
detailed models, respectively.

\begin{figure*}[!tb]
    \centering
    {\raggedright\normalfont High-level to phenomenological: $\mathrm{GP}_{\text{hp}}$\par\smallskip}
    \GeneHPGPTrainingScreening
    \par\medskip
    {\raggedright\normalfont High-level to mechanistic: $\mathrm{GP}_{\text{hm}}$\par\smallskip}
    \GeneGPTrainingScreening
    \caption{High-level to phenomenological (top) and mechanistic (bottom)
    learning and filtering with $\mathrm{GP}_{\text{hp}}$ and
    $\mathrm{GP}_{\text{hm}}$, respectively. In each row, (a) balanced
    test error compares GP-MPM and uniform random
    sampling on the same 200 independent test points through 200
    acquisitions. (b,c) Filtering at $\zeta=0.5$ and $\zeta=0.8$ on
    the same 1000 proposals for each model. The horizontal axis
    shows the GP mean of $\epsilon-d_{(\lceil\gamma M\rceil)}(\theta)$.
    Histogram weights are percentages of the original proposals.}
    \label{fig:hp-gp-diagnostics}
    \label{fig:ex2-gp-diagnostics}
\end{figure*}

\subsubsection{Behavioral mapping and GP-guided inference}
For behavioral comparisons, we draw $M=3200$ joint parameter tuples from the final $5000$ reduced model ensemble steps and simulate their GFP trajectories at all six DNA conditions. For $d(\cdot, \cdot)$, we use the
mean of the six condition-wise RMSEs as the trajectory discrepancy measure. For the hyperparameters, we choose
$\epsilon=0.5\,\mu\mathrm{M}$ (bounds the discrepancy between detailed and high-level model predictions), $\gamma=0.5$ (the minimum behavioral support required for a detailed parameter tuple) and $\zeta=0.8$ (the threshold for filtering MCMC proposals using the trained GP). For GP training, we use $\epsilon-d_{(\lceil\gamma M\rceil)}(\theta)$ as
the target, where $d_{(r)}(\theta)$ is the $r$-th smallest discrepancy
$d(y(\cdot;\theta),\hat y(\cdot;\hat\theta^{(j)}))$ over the reduced
posterior samples. This quantity is non-negative when
$\hat b_{\epsilon,M}(\theta)>\gamma$. Figure~\ref{fig:ex2-gp-diagnostics}(a) shows the
balanced errors at $\zeta=0.5$ through $200$ acquisitions on a common $200$-point
independent test set for each sampling strategy while training the GP. For $\mathrm{GP}_{\text{hp}}$, the final
balanced errors are $\GeneHPGPBalancedErrorPct\%$ for GP-MPM and
$\GeneHPRandomBalancedErrorPct\%$ for uniform random sampling. At
$\zeta=0.8$, $\mathrm{GP}_{\text{hp}}$ retains
$\GeneHPGPCompatibleRetainedPct\%$ of compatible test points and rejects
$\GeneHPGPIncompatibleRejectedPct\%$ of incompatible test points. For $\mathrm{GP}_{\text{hm}}$, the balanced errors are
$\GeneGPBalancedErrorPct\%$ for GP-MPM and
$\GeneRandomBalancedErrorPct\%$ for uniform random sampling. At
$\zeta=0.8$, $\mathrm{GP}_{\text{hm}}$ retains
$\GeneGPCompatibleRetainedPct\%$ of compatible test points and rejects
$\GeneGPIncompatibleRejectedPct\%$ of incompatible test points. At $\zeta=0.8$, for stricter filtering, $\mathrm{GP}_{\text{hp}}$ has a balanced error of $\GeneHPGPStrictBalancedErrorPct\%$ and $\mathrm{GP}_{\text{hm}}$ has a balanced error of $\GeneGPStrictBalancedErrorPct\%$. We use the trained GP as a filter for MCMC proposals before likelihood evaluation, see Figure~\ref{fig:ex2-gp-diagnostics}(b,c) for illustrations of the filtering at two different levels. The stricter filtering only accepts proposals with a high probability of being compatible. 

\subsubsection{Inferred trajectories and parameter distributions}
\begin{figure}[bt]
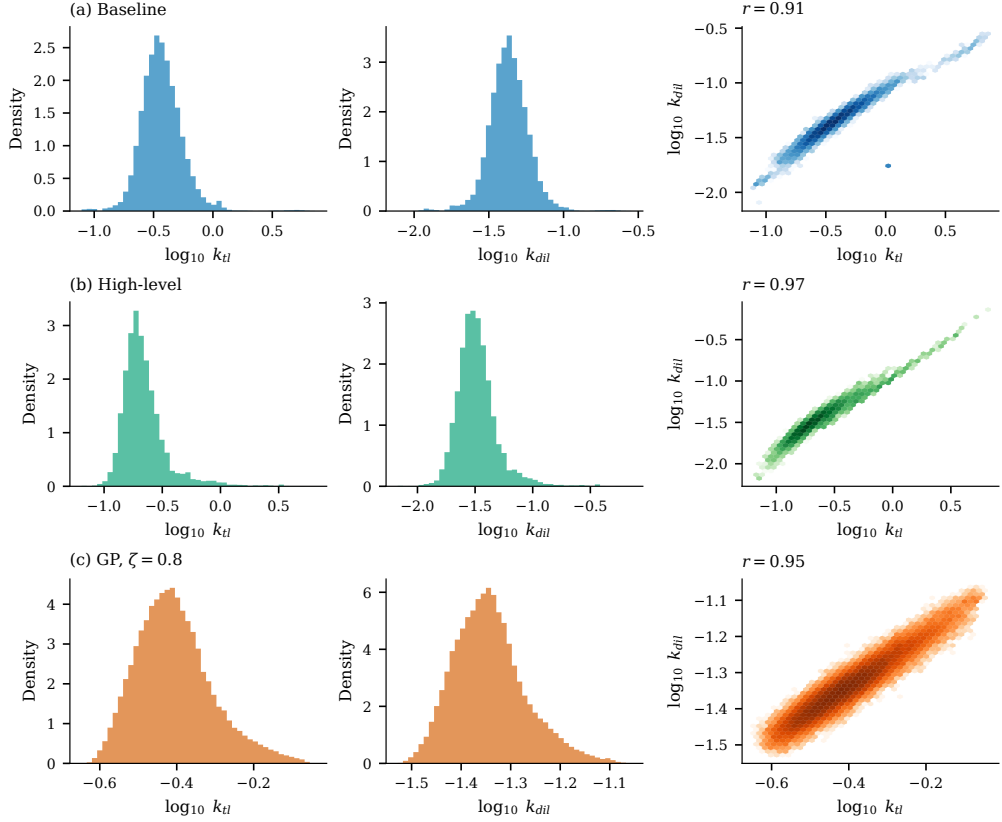

    \centering
    \GeneParameterDependence
    \caption{Selected parameter marginals and joint distributions in
    $\log_{10}$ coordinates. (a) Detailed baseline: \GeneDetailedPairFirst and \GeneDetailedPairSecond. (b) High-level model: \GeneReducedPairFirst \:
    and \GeneReducedPairSecond. (c) The detailed-model pair under GP-guided
    inference. Pearson correlations use the displayed posterior windows.}
    \label{fig:ex2-baseline-parameters}
\end{figure}
From the posterior distributions of the baseline and GP-guided inference of both the phenomenological and mechanistic models, we compute the mean condition-wise MAP RMSEs. For the phenomenological model, we get $\GeneHPBaselineRMSE\,\mu\mathrm{M}$ for the baseline and $\GeneHPGuidedRMSE\,\mu\mathrm{M}$ for the GP-guided inference. For the mechanistic model, we get $\GeneBaselineRMSE\,\mu\mathrm{M}$ and $\GeneGuidedRMSE\,\mu\mathrm{M}$, respectively. The corresponding predictions at each DNA concentration are shown in Figure~\ref{fig:ex2-guided-fits}(a,b), respectively. The posterior marginal distributions for the parameters are shown in Figure~\ref{fig:ex2-guided-marginals}, for both the phenomenological and the mechanistic models for each inferred parameter and in Figure~\ref{fig:ex2-baseline-parameters}, we show three
selected parameter marginals and their joint distributions for baseline, high-level and guided inference. These results show that the GP-guided inference retains similar parameter posterior distributions for some parameters, while for some others the posteriors are narrower. The inference also reduces the number of likelihood evaluations and the overall inference time.  

\begin{figure*}[!tb]
    \centering
    \begin{minipage}[t]{0.488\textwidth}
        \raggedright\normalfont (a) Phenomenological model\par\smallskip
        \GeneHPBaselineGuidedFits
    \end{minipage}\hfill
    \begin{minipage}[t]{0.488\textwidth}
        \raggedright\normalfont (b) Detailed mechanistic model\par\smallskip
        \GeneBaselineGuidedFits
    \end{minipage}
    \caption{Baseline and GP-guided GFP trajectories at six DNA
    concentrations for (a) the phenomenological model with
    $\mathrm{GP}_{\text{hp}}$ and (b) the detailed mechanistic model with
    $\mathrm{GP}_{\text{hm}}$. Each panel shows experimental means and SD,
    the sampled MAP and 20 joint posterior draws from each inference run.}
    \label{fig:hp-guided-fits}
    \label{fig:ex2-guided-fits}
\end{figure*}

\begin{figure}[t]
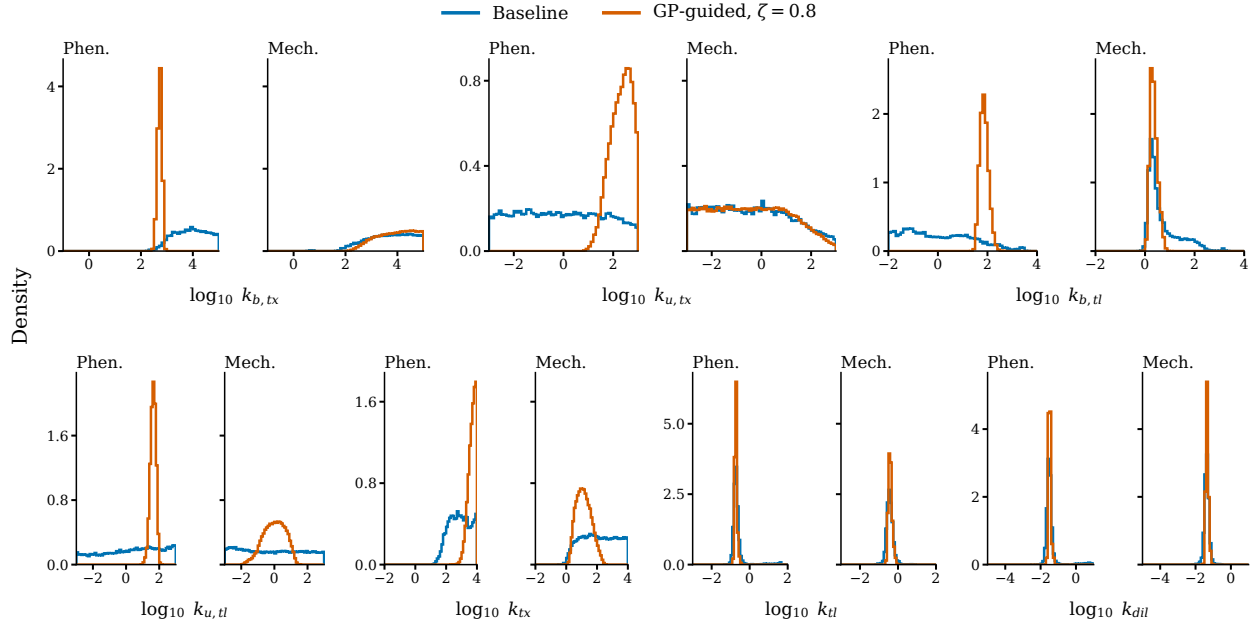

    \centering
    \GeneHierarchyMarginals
    \caption{Baseline and GP-guided marginal distributions for seven
    inferred parameters. Each group pairs the phenomenological
    model (left) with the mechanistic model (right). Blue and orange
    curves show baseline and guided densities, respectively, using the
    final 10000 ensemble steps of each run.}
    \label{fig:ex2-guided-marginals}
\end{figure}
\begin{table}[!htbp]
    \centering
    \small
    \caption{Baseline and GP-guided inference for both model levels.
    Times are in minutes. Runtime reductions use baseline MCMC elapsed
    time as the denominator. Autocorrelation times alone do not establish
    convergence for these models and are not reported here.}
    \label{tab:ex2-performance}
    \begin{tabular}{lrrrr}
\toprule
& \multicolumn{2}{c}{Phenomenological} & \multicolumn{2}{c}{Mechanistic}\\
Metric & Baseline & $\mathrm{GP}_{\text{hp}}$ & Baseline & $\mathrm{GP}_{\text{hm}}$\\
\midrule
Model MCMC (min) & \GeneHPBaselineMinutes & \GeneHPGuidedMinutes & \GeneBaselineMinutes & \GeneGuidedMinutes\\
MCMC time reduction (\%) & 0 & \GeneHPMCMCTimeReductionPct & 0 & \GeneMCMCTimeReductionPct\\
Mean MAP RMSE ($\mu$M) & \GeneHPBaselineRMSE & \GeneHPGuidedRMSE & \GeneBaselineRMSE & \GeneGuidedRMSE\\
\midrule
\multicolumn{5}{l}{One-time setup costs (min)}\\
Reduced inference & -- & \GeneHPReducedMinutes & -- & \GeneReducedMinutes\\
GP training & -- & \GeneHPGPTrainingMinutes & -- & \GeneGPTrainingMinutes\\
\bottomrule
\end{tabular}

\end{table}
\subsubsection{Computational performance}
Each inference runs on the UC Merced High Performance Computing (HPC) cluster with 32 Intel Xeon CPU workers and 16 GB memory. The phenomenological baseline MCMC
takes $\GeneHPBaselineMinutes$ minutes and its guided run takes
$\GeneHPGuidedMinutes$ minutes, a $\GeneHPMCMCTimeReductionPct\%$ reduction. We observe a much more significant speedup in the more complex mechanistic model inference. For this, 
MCMC takes $\GeneBaselineMinutes$ minutes and its guided run takes
$\GeneGuidedMinutes$ minutes, speeding up the inference by $\GeneMCMCTimeReductionPct\%$. The $\mathrm{GP}_{\text{hp}}$ filter rejects $\GeneHPGPProposalRejectedPct\%$
of proposals that pass the ordinary priors and takes
$\GeneHPGPScreeningMinutes$ minutes. The corresponding values for
$\mathrm{GP}_{\text{hm}}$ are $\GeneGPProposalRejectedPct\%$ and
$\GeneGPScreeningMinutes$ minutes. Reduced model inference and GP training are one-time setup costs, taking $\GeneReducedMinutes$ minutes for the reduced model and $\GeneGPTrainingMinutes$ minutes for the GP training.
Table~\ref{tab:ex2-performance} reports MCMC times, percentage time
reductions, fit errors and the separate one-time setup costs.

\FloatBarrier

\section{Discussion}
\label{sec:discussion}
One of the notable aspects of our method is its ability to speed up the inference process in non-identifiable models, where standard approaches tend to spend the computational budget on exploring ranges of parameters that
produce similar measured trajectories. Exploring distant parts of these
families can require many detailed simulations but with similar predictive
results. The behavioral GP addresses this by identifying a region compatible with the higher-level posterior and focuses detailed inference there. Although this comes at a cost: the identified parameter combinations may not be implementable. Therefore, it is important to train the GP using a region of parameters that the designer considers relevant. In the
gene expression example, we show this by
allocating the computation to explore parameter combinations that realize the behavior of interest in a region of interest. Indeed, in Figure~\ref{fig:ex2-guided-marginals} we show that behavioral filtering concentrates some marginals while others remain broad. In the mechanistic model inference, the guided marginals of $k_{tx}$ and
$k_{u,tl}$ are narrower, while those of $k_{b,tx}$ and $k_{u,tx}$ retain
broad ranges. The strong positive association between $k_{tl}$ and
$k_{dil}$ persists in the baseline and guided chains. Clearly, in the GP-guided inference we have an exploration-speedup tradeoff: the GP filter concentrates the inference on a smaller region of parameter space, which speeds up the inference, but it also reduces exploration of the full range of compatible parameters. In addition to designing the compatibility criterion, our framework allows the user to choose various tolerances: the compatibility tolerance $\epsilon$ between the detailed and high-level model predictions, the required high-level posterior compatibility mass, $\gamma$, and the strictness of the GP filtering, $\zeta$. So these choices give the user control
over the balance between parameter exploration and behavioral focus.

Since our approach does not require any parameter or state correspondence, we are able to explore the parameter inference problem through successive model levels in a hierarchy. The key requirement is that all models in the hierarchy share a common measured output, which provides a route for transferring information through an extended modeling hierarchy by building approximate simulation relations between these models. A detailed model posterior, therefore, can aid the inference of an even more
detailed model. Future work will develop this procedure across successive levels, introducing additional mechanisms and parameters at each level and learning a new behavioral mapping between the chosen models. This
would allow information obtained at tractable levels to guide the
identification of increasingly detailed models. As simulation costs grow,
avoiding repeated detailed evaluations becomes increasingly valuable.
The approach offers a route toward parameter inference for models whose
computational demands currently limit their practical use.

\section{Conclusion}    
We presented a behavioral mapping framework for accelerating Bayesian inference across hierarchical dynamical models without requiring parameter correspondence. The method transfers information through posterior behavioral compatibility and a learned Gaussian Process (GP) support level set. For nonlinear gene expression, the application connects a two-state high-level model, a nine-state phenomenological model, and a 38-state mechanistic model. Behavioral guidance reduces MCMC time by \GeneHPMCMCTimeReductionPct\% for the phenomenological model and \GeneMCMCTimeReductionPct\% for the mechanistic model. The presented framework formalizes behavioral compatibility between hierarchies of models and is readily applicable for many other systems.

The code for both applications is available at
\url{https://github.com/ctsd-lab/hierarch-inf}.
\section*{Acknowledgments}
Part of this research was conducted using the Pinnacles cluster, which is centrally funded by the University of California, Merced, and maintained by the Cyberinfrastructure and Research Technologies (CIRT) team at UC Merced.
\subsection*{AI Use Statement}
The authors used several AI tools for brainstorming, assistance with the computational implementations, data visualizations, and  writing of the manuscript. AI generated outputs were reviewed, vetted, and edited by the authors for accuracy, and the authors take full responsibility for the content of the manuscript.  
\bibliography{refs.bib}
\bibliographystyle{IEEEtran}

\end{document}